%% file: manuscript.tex
\documentclass[11pt,  draftclsnofoot, onecolumn]{IEEEtran}
\usepackage[latin9]{inputenc}
\usepackage{color}
\usepackage{rotfloat}
\usepackage{amsthm}
\usepackage{amsmath}
\usepackage{amssymb}
\usepackage{graphicx}
\usepackage{esint}
\usepackage{xargs}[2008/03/08]
\usepackage[unicode=true,
 bookmarks=true,bookmarksnumbered=false,bookmarksopen=false,
 breaklinks=true,pdfborder={0 0 0},backref=false,colorlinks=true]
 {hyperref}

\makeatletter
  \theoremstyle{definition}
  \newtheorem{defn}{\protect\definitionname}
\theoremstyle{plain}
\newtheorem{thm}{\protect\theoremname}
  \theoremstyle{plain}
  \newtheorem{lem}{\protect\lemmaname}

\usepackage[algo2e,ruled,titlenotnumbered]{algorithm2e}
\usepackage{cite}

\makeatother

  \providecommand{\definitionname}{Definition}
  \providecommand{\lemmaname}{Lemma}
\providecommand{\theoremname}{Theorem}

\begin{document}

\title{Robust 1-bit Compressive Sensing via Gradient Support Pursuit}

\author{Sohail~Bahmani,~\IEEEmembership{Student,~IEEE}%
\thanks{S.B. is with Department of Electrical and Computer Engineering, Carnegie
Mellon University, 5000 Forbes Ave., Pittsburgh, PA, 15213 e-mail:\texttt{ sbahmani@cmu.edu}%
}, Petros~T.~Boufounos,~\IEEEmembership{Member,~IEEE}%
\thanks{P.T.B. is with Mitsubishi Electric Research Laboratory, 201 Broadway,
Cambridge, MA 02139 e-mail:\texttt{ petrosb@merl.com}%
}, Bhiksha~Raj%
\thanks{B.R. is with Language Technologies Institute, Carnegie Mellon University,
5000 Forbes Ave., Pittsburgh, PA, 15213 e-mail: \texttt{bhiksha@cs.cmu.edu}%
}}

\maketitle
\IEEEpeerreviewmaketitle
\begin{abstract}
This paper studies a formulation of 1-bit Compressive Sensing (CS)
problem based on the maximum likelihood estimation framework. In order
to solve the problem we apply the recently proposed Gradient Support
Pursuit algorithm, with a minor modification. Assuming the proposed
objective function has a Stable Restricted Hessian, the algorithm
is shown to accurately solve the 1-bit CS problem. Furthermore, the
algorithm is compared to the state-of-the-art 1-bit CS algorithms
through numerical simulations. The results suggest that the proposed
method is robust to noise and at mid to low input SNR regime it achieves
the best reconstruction SNR vs. execution time trade-off.\end{abstract}
\input{macros.tex}

\section{Introduction\label{sec:1bit-Intro}}

\IEEEPARstart{Q}{uantization} is an indispensable part of digital
signal processing and digital communications systems. To incorporate
Compressive Sensing (CS) methods in these systems, it is thus necessary
to analyze and evaluate them considering the effect of measurement
quantization. In the recent years there is a growing interest in quantized
CS in the literature \cite{Laska_Finite_2009,Dai_Distortion-rate_2009,Sun_Optimal_2009,Zymnis_Compressed_2010,Jacques_Dequantizing_2011,Laska_Democracy_2011},
particularly the extreme case of quantization to a single bit, known
as 1-bit Compressive Sensing \cite{Boufounos_1-bit_2008}, where only
the sign of the linear measurements are recorded. The advantage of
this acquisition scheme is that it can be implemented using simple
hardware that is not expensive and can operate at very high sampling
rates. The formulation of this problem is also very similar to the
sparse logistic regression, very useful in machine learning applications
\cite{Plan_Robust_2013}.

As in standard CS, the algorithms proposed for the 1-bit CS problem
can be categorized into convex methods and non-convex greedy methods.
In \cite{Boufounos_1-bit_2008} an algorithm is proposed for 1-bit
CS that induces sparsity through the $\ell_{1}$-norm while penalizes
inconsistency with the 1-bit sign measurements via a convex regularization
term. In a noise-free scenario, the 1-bit measurements do not convey
any information about the length of the signal. Therefore, the algorithm
of \cite{Boufounos_1-bit_2008}, as well as other 1-bit CS algorithms,
aim at accurate estimation of the normalized signal. Requiring the
1-bit CS estimate to lie on the surface of the unit-ball imposes a
non-convex constraint in methods that perform an (approximate) optimization,
even those that use the convex $\ell_{1}$-norm to induce sparsity.
Among greedy 1-bit CS algorithms, an algorithm called Matching Sign
Pursuit (MSP) is proposed in \cite{Boufounos_Greedy_2009} based on
the CoSaMP algorithm \cite{Needell_CoSaMP_2009}. This algorithm is
empirically shown to perform better than the standard CoSaMP algorithm
for estimation of the normalized sparse signal. In \cite{Laska_Trust_2011}
the Restricted-Step Shrinkage (RSS) algorithm is proposed for 1-bit
CS problems, improving the performance of $\ell_{1}$-based algorithms.
This algorithm, which is similar to \emph{trust-region} algorithms
in non-convex optimization, is shown to converge to a stationary point
of the objective function regardless of the initialization. More recently,
\cite{Jacques_Robust_2013} derived a lower bound on the best achievable
reconstruction error of any 1-bit CS algorithm in noise-free scenarios.
Furthermore, using the notion of ``binary stable embeddings'', they
have shown that Gaussian measurement matrices can be used for 1-bit
CS problems both in noisy and noise-free regime. The Binary Iterative
Hard Thresholding (BIHT) algorithm is also proposed in \cite{Jacques_Robust_2013}
and shown to have favorable performance compared to the RSS and MSP
algorithms through numerical simulations. For robust 1-bit CS in presence
of noise, \cite{Yan_Robust_2012} also proposed the Adaptive Outlier
Pursuit (AOP) algorithm. In each iteration of the AOP, first the sparse
signal is estimated similar to BIHT with the difference that the potentially
corrupted measurements are excluded. Then with the new signal estimate
fixed, the algorithm updates the list of likely corrupted measurements.
The AOP is shown to improve on performance of BIHT through numerical
simulations. \cite{Plan_One-bit_2011} proposed a linear program to
solve the 1-bit CS problems in a noise-free scenario. The algorithm
is proved to provide accurate solutions, albeit using a sub-optimal
number of measurements. Furthermore, in \cite{Plan_Robust_2013} a
convex program is proposed that is robust to noise in 1-bit measurements
and achieves the optimal number of measurements and in \cite{Ai_One-bit_NG_2013}
it is shown that non-Gaussian matrices can be used for acquisition
under certain conditions on the acquired signal.

In this paper, we formulate the 1-bit CS problem as an sparsity-constrained
optimization. As described in Section \ref{sec:1bit-Formulation},
the objective function is obtained by adjusting the loss function
that arises in the Maximum Likelihood Estimation (MLE) formulation
of the problem. To solve this optimization problem in Section \ref{sec:1bit-Algorithm}
we propose a slightly modified version of the Gradient Support Pursuit
(GraSP) algorithm proposed in \cite{Bahmani_Greedy_2012}. In Section
\ref{sec:1bit-AccyGrnty}, the algorithm is shown to yield an approximate
solution provided that the objective function satisfies certain sufficient
conditions. Furthermore, in Section \ref{sec:Simulations}, we compare
the performance of our algorithm with the BIHT algorithm and the non-convex
variant of the 1-bit CS solver introduced in \cite{Plan_Robust_2013}
through numerical simulations. As an aside, we show that the non-convex
solver described in \cite{Plan_Robust_2013} has an explicit solution
(see Appendix \ref{apndx:PV0}). Finally, we discuss and conclude
in Section \ref{sec:Conclusion}.

\paragraph*{Notation}

In the remainder of the paper we denote the positive part of a real
number $x$ by $\left(x\right)_{+}$. For a positive integer $k$,
the set $\left\{ 1,2,\ldots,k\right\} $ is denoted by $\left[k\right]$.
The indicator function is denoted by $\bbone\left(\cdot\right)$.Vectors
and matrices are denoted by boldface characters. With the exception
of $\st N$ which we use to denote the normal distribution as in $\st N\left(0,1\right)$,
calligraphic letters are reserved for denoting sets. The support set
(i.e., the set of non-zero coordinates) of a vector $\vc x$ is denoted
by $\supp\left(\vc x\right)$. The best $k$-term approximation of
a vector $\vc v$ is denoted by $\vc v_{k}$. Depending on the context
$\vc v_{\left[k\right]}$ may also be a $k$-dimensional vector denoting
the restriction of $\vc v$ to its $k$ largest entries in magnitudes
(i.e., truncated best $k$-term approximation of $\vc v$). Restriction
of an $n$-dimensional vector $\vcg v$ to its entries corresponding
to an index set $\st I\subseteq\left[n\right]$ is denoted by $\res{\vc v}_{\st I}$.
Similarly $\mx A_{\st I}$ denotes the restriction of a matrix $\mx A$
to the columns enumerated by $\st I$. Restriction of the identity
matrix to the columns enumerated by $\st I$ is particularly denoted
by $\mx P_{\st I}$. The operator norm of $\mx A$ with respect to
the $\ell_{2}$-norm is denoted by $\norm A$. For square matrices
$\mx A$ and $\mx B$ we write $\mx B\preccurlyeq\mx A$ to state
that $\mx A-\mx B$ is positive semidefinite. We use $\nabla f\left(\cdot\right)$
and $\nabla^{2}f\left(\cdot\right)$ to denote the gradient and the
Hessian of a twice continuously differentiable function $f:\mathbb{R}^{n}\mapsto\mathbb{R}$.
For an index set $\st I\subset\left[n\right]$, the restriction of
the gradient to the entries selected by $\st I$ and the restriction
of the Hessian to the entries selected by $\st I\times\st I$ are
denoted by $\nabla_{\st I}f\left(\cdot\right)$ and $\nabla_{\st I}^{2}f\left(\cdot\right)$,
respectively. Finally, numerical superscripts within parentheses denote
the iteration index.

\section{Problem Formulation\label{sec:1bit-Formulation}}

We cast the 1-bit CS problem in the framework of statistical parametric
estimation which is also considered in \cite{Zymnis_Compressed_2010}.
In 1-bit CS, binary measurements $y\in\left\{ \pm1\right\} $ of a
signal $\vc x^{\star}\in\mathbb{R}^{n}$ are collected based on the
model 
\begin{alignat}{1}
y & =\text{sgn}\left(\left\langle \vc a,\vc x^{\star}\right\rangle +e\right),\label{eq:Model}
\end{alignat}
 where $\vc a$ is a measurement vector and $e$ denotes an additive
noise with distribution $\st N\negmedspace\left(0,\!\sigma^{2}\right)$.
It is straightforward to show the conditional likelihood of $y$ given
$\vc a$ and signal $\vc x$ can be written as
\begin{alignat*}{1}
\Pr\left\{ y\mid\vc a;\vc x\right\}  & =\Phi\left(y\frac{\left\langle \vc a,\vc x\right\rangle }{\sigma}\right),
\end{alignat*}
 with $\Phi\left(\cdot\right)$ denoting the standard normal cumulative
distribution function (CDF). Then, for measurement pairs $\left\{ \left(\vc a_{i},y_{i}\right)\right\} _{i=1}^{m}$the
MLE loss function is given by 
\begin{alignat*}{1}
f_{\mbox{\tiny{MLE}}}\left(\vc x\right) & :=-\frac{1}{m}\sum_{i=1}^{m}\log\left(\Phi\left(y_{i}\frac{\left\langle \vc a_{i},\vc x\right\rangle }{\sigma}\right)\right).
\end{alignat*}
 The estimator should exploit the sparsity of the solution and sparsely
minimize the function above. Note, however, that at high Signal-to-Noise
Ratio (SNR) this function does not lend itself easily to optimization.
To observe this behavior, rewrite $f_{\tiny{\mbox{MLE}}}$ as 
\begin{alignat*}{1}
f_{\tiny{\mbox{MLE}}}\left(\vc x\right) & =\frac{1}{m}\sum_{i=1}^{m}g_{\eta}\left(y_{i}\left\langle \vc a_{i},\frac{\vc x}{\norm[2]{\vc x^{\star}}}\right\rangle \right),
\end{alignat*}
where $\eta:=\frac{\norm[2]{\vc x^{\star}}}{\sigma}$ is the SNR and
$g_{\omega}\left(t\right):=-\log\Phi\left(\omega t\right)$ for all
$\omega\geq0$. As $\eta\to+\infty$ the function $g_{\eta}\left(t\right)$
tends to
\begin{alignat*}{1}
g_{\infty}\left(t\right) & :=\begin{cases}
0 & t>0\\
\log2 & t=0\\
+\infty & t<0
\end{cases}.
\end{alignat*}
 Therefore, as the SNR increases to infinity $f_{\mbox{\tiny{MLE}}}\left(\vc x\right)$
tends to a sum of discontinuous constant functions that do not uniquely
identify the solution and are difficult to handle in practice. This
is essentially the same problem as the amplitude ambiguity demonstrated
in the original formulation in \cite{Boufounos_1-bit_2008}. Furthermore,
whether the noise level is too low or the signal is too strong relative
to the noise, in a high (but finite) SNR scenario the measurement
vectors are likely to be consistent with the noise-free measurements
of the true signal $\vc x^{\star}$. In these cases, $f_{\tiny{\mbox{MLE}}}\left(t\vc x^{\star}\right)$
can be made arbitrarily close to the zero lower bound as $t\to+\infty$.
Therefore, $f_{\tiny{\mbox{MLE}}}$ would not have a bounded minimizer.
This can be interpreted as an infinite estimation error.

To avoid the problems mentioned above we consider a modified loss
function

\begin{alignat}{1}
f_{0}\left(\vc x\right) & :=-\frac{1}{m}\sum_{i=1}^{m}\log\left(\Phi\left(y_{i}\left\langle \vc a_{i},\vc x\right\rangle \right)\right),\label{eq:1bitLoss}
\end{alignat}
while we merely use an alternative formulation of (\ref{eq:Model})
given by 
\begin{alignat*}{1}
y & =\text{sgn}\left(\eta\left\langle \vc a,\vc x^{\star}\right\rangle +e\right),
\end{alignat*}
in which $\eta>0$ denotes the true SNR, $\vc x^{\star}$ is assumed
to be unit-norm, and $e\sim\st N\left(0,1\right)$. The aim is accurate
estimation of the unit-norm signal $\vc x^{\star}$ which is assumed
to be $s$-sparse. Disregarding computational complexity, the candidate
estimator would be
\begin{alignat}{1}
\arg\min_{\vc x} & \ f_{0}\left(\vc x\right)\quad\text{s.t. }\norm[0]{\vc x}\leq s\text{ and }\norm[2]{\vc x}\leq1.\label{eq:BMLE}
\end{alignat}
However, finding the exact solution (\ref{eq:BMLE}) may be computationally
intractable, thereby we merely focus on approximate solutions to this
optimization problem.

\section{Algorithm\label{sec:1bit-Algorithm}}

In this section we introduce a modified version of the GraSP algorithm,
outlined in Algorithm \ref{alg:1bit-GraSP}, for estimation of bounded
sparse signals associated with a cost function. While in this paper
the main goal is to study the 1-bit CS problem and in particular the
objective function described by (\ref{eq:1bitLoss}), we state performance
guarantees of Algorithm \ref{alg:1bit-GraSP} in more general terms.
As in GraSP, in each iteration first the $2s$ coordinates at which
the gradient of the cost function at the iterate $\vc x\itr t$ has
the largest magnitudes are identified. These coordinates, denoted
by $\st Z$, are then merged with the support set of $\vc x\itr t$
to obtain the set $\st T$ in the second step of the iteration. Then,
as expressed in line \ref{lin:crude} of Algorithm \ref{alg:1bit-GraSP},
a crude estimate $\vc b$ is computed by minimizing the cost function
over vectors of length no more than $r$ whose supports are subsets
of $\st T$. Note that this minimization would be a convex program
and therefore tractable, provided that the sufficient conditions proposed
in Section \ref{sec:1bit-AccyGrnty} hold. In the final step of the
iteration (i.e., line \ref{lin:prune}) the crude estimate is pruned
to its best $s$-term approximation to obtain the next iterate $\vc x\itr{t+1}$.
By definition we have $\norm[2]{\vc b}\leq r$, thus the new iterate
remains in the feasible set (i.e., $\norm[2]{\vc x\itr{t+1}}\leq r$).

{\centering
\begin{algorithm2e}
	\DontPrintSemicolon
	\caption{GraSP with Bounded Thresholding\label{alg:1bit-GraSP}}
	\SetKwInOut{Input}{input}
	\SetKwInOut{Output}{output}
	\Input{\begin{tabular}{ll}
			$s$ & desired sparsity level\\
			$r$ & radius of the feasible set\\
			$f\left(\cdot\right)$ & the cost function\\
			\end{tabular}}
	$t\longleftarrow 0$\;
	$\vc{x}\itr{t}\longleftarrow \vc{0}$\;
	\Repeat{halting condition holds}{
\nl\label{lin:identify}	$\st{Z} \longleftarrow \supp\left(\left[\nabla f\left(\vc{x}\itr{t}\right)\right]_{2s}\right)$\;
\nl\label{lin:merge}	$\st{T} \longleftarrow \supp\left(\vc{x}\itr{t}\right)\cup\st{Z}$\;
\nl\label{lin:crude}	$\vc{b} \longleftarrow \displaystyle{\arg\min_{\vc{x}}}\ f\left(\vc{x}\right)\quad\text{s.t. }\res{\vc{x}}_{\st{T}\cmpl}=\vc{0} \text{ and } \norm[2]{\vc{x}}\leq r$\;    
\nl\label{lin:prune}	$\vc{x}\itr{t+1} \longleftarrow \vc{b}_{s}$\;
\nl		$t\longleftarrow t+1$\;
	}
	\Return{$\vc{x}\itr{t}$}
\end{algorithm2e}}

\section{Accuracy Guarantees\label{sec:1bit-AccyGrnty}}

In order to provide accuracy guarantees for Algorithm \ref{alg:1bit-GraSP},
we rely on the notion of SRH described in \cite{Bahmani_Greedy_2012}
with a slight modification in its definition. The original definition
of SRH basically characterizes the cost functions that have bounded
curvature over sparse canonical subspaces, possibly at locations arbitrarily
far from the origin. However, we only require the bounded curvature
condition to hold at locations that are within a sphere around the
origin. More precisely, we redefine the SRH as follows.
\begin{defn}[Stable Restricted Hessian]
\label{def:SRH-bnd} Suppose that $f:\mathbb{R}^{n}\mapsto\mathbb{R}$
is a twice continuously differentiable function and let $k<n$ be
a positive integer. Furthermore, let $\alpha_{k}\left(\vc x\right)$
and $\beta_{k}\left(\vc x\right)$ be in turn the largest and smallest
real numbers such that 
\begin{alignat}{2}
\beta_{k}\left(\vc x\right)\norm[2]{\vcg{\Delta}}^{2} & \leq\vcg{\Delta}\tran\nabla^{2}f\left(\vc x\right)\vcg{\Delta} & \leq\alpha_{k}\left(\vc x\right)\norm[2]{\vcg{\Delta}}^{2},\label{eq:SRH-1}
\end{alignat}
 holds for all $\vcg{\Delta}$ and $\vc x$ that obey $\left|\supp\left(\vcg{\Delta}\right)\cup\supp\left(\vc x\right)\right|\leq k$
and $\norm[2]{\vc x}\leq r$. Then $f$ is said to have an Stable
Restricted Hessian of order $k$ with constant $\mu_{k}\geq1$ in
a sphere of radius $r>0$, or for brevity $\left(\mu_{k},r\right)$-SRH,
if $1\leq\alpha_{k}\left(\vc x\right)/\beta_{k}\left(\vc x\right)\leq\mu_{k}$
for all $k$-sparse $\vc x$ with $\norm[2]{\vc x}\leq r$.\end{defn}
\begin{thm}
\label{thm:IterationInvariant}Let $\overline{\vc x}$ be a vector
such that $\norm[0]{\overline{\vc x}}\leq s$ and $\norm[2]{\overline{\vc x}}\leq r$.
If the cost function $f\left(\vc x\right)$ have $\left(\mu_{4s},r\right)$-SRH
corresponding to the curvature bounds $\alpha_{4s}\left(\vc x\right)$
and $\beta_{4s}\left(\vc x\right)$ in (\ref{eq:SRH-1}), then iterates
of Algorithm \ref{alg:1bit-GraSP} obey 
\begin{alignat*}{1}
\norm[2]{\vc x\itr{t+1}-\overline{\vc x}} & \leq\left(\mu_{4s}^{2}-\mu_{4s}\right)\norm[2]{\vc x\itr t-\overline{\vc x}}+2\left(\mu_{4s}+1\right)\epsilon,
\end{alignat*}
 where $\epsilon$ obeys $\norm[2]{\left[\nabla f\left(\overline{\vc x}\right)\right]_{3s}}\leq\epsilon\,\beta_{4s}\left(\vc x\right)$
for all $\vc x$ with $\norm[0]{\vc x}\leq4s$ and $\norm[2]{\vc x}\leq r$.
\end{thm}
The immediate implication of this theorem is that if the 1-bit CS
loss $f_{0}\left(\vc x\right)$ has $\left(\mu_{4s},1\right)$-SRH
with $\mu_{4s}\leq\frac{1+\sqrt{3}}{2}$ then we have $\norm[2]{\vc x\itr t-\vc x^{\star}}\leq2^{-t}\norm[2]{\vc x^{\star}}+2\left(3+\sqrt{3}\right)\epsilon$. 

Proof of Theorem \ref{thm:IterationInvariant} is almost identical
to the proof of Theorem 1 in \cite{Bahmani_Greedy_2012}. For brevity
we will provide a proof sketch in Appendix \ref{apndx:1-bitCS} and
elaborate only on the more distinct parts of the proof and borrow
the remaining parts from \cite{Bahmani_Greedy_2012}.

\section{Simulations\label{sec:Simulations}}

In our simulations using synthetic data we considered signals of dimensionality
$n=1000$ that are $s$-sparse with $s=10,20,$ or $30$. The non-zero
entries of the signal constitute a vector randomly drawn from the
surface of the unit Euclidean ball in $\mathbb{R}^{s}$. The $m\times n$
measurement matrix has iid standard Gaussian entries with $m$ varying
between 100 and 2000 in steps of size 100. We also considered three
different noise variances $\sigma^{2}$ corresponding to input SNR
$\eta=20$dB, $10$dB, and $0$dB. Figures \ref{fig:AE}--\ref{fig:ET}
illustrate the average performance of the considered algorithm over
200 trials versus the sampling ratio (i.e., $m/n$). In these figures,
the results of Algorithm \ref{alg:1bit-GraSP} considering $f_{0}$
and $f_{\tiny{MLE}}$ as the objective function are demarcated by
GraSP and GrasP-$\eta$, respectively. Furthermore, the results corresponding
to BIHT algorithm with one-sided $\ell_{1}$ and $\ell_{2}$ objective
functions are indicated by BIHT and BIHT-$\ell_{2}$, respectively.
We also considered the $\ell_{0}$-constrained optimization proposed
by \cite{Plan_Robust_2013} which we refer to as PV-$\ell_{0}$. While
\cite{Plan_Robust_2013} mostly focused on studying the convex relaxation
of this method using $\ell_{1}$-norm, as shown in Appendix \ref{apndx:PV0}
the solution to PV-$\ell_{0}$ can be derived explicitly in terms
of the one-bit measurements, the measurement matrix, and the sparsity
level. We do not evaluate the convex solver proposed in \cite{Plan_Robust_2013}
because we did not have access to an efficient implementation of this
method. Furthermore, this convex solver is expected to be inferior
to PV-$\ell_{0}$ in terms of accuracy because it operates on a feasible
set with larger \emph{mean width} \cite[Theorem 1.1]{Plan_Robust_2013}.
With the exception of the non-iterative PV-$\ell_{0}$, the other
four algorithms considered in our simulations are iterative; they
are configured to halt when they produce an estimate whose 1-bit measurements
and the real 1-bit measurements have a Hamming distance smaller than
an $\eta$-dependent threshold.

Figure \ref{fig:AE} illustrates performance of the considered algorithms
in terms of the angular error between the normalized estimate $\widehat{\vc x}$
and the true signal $\vc x^{\star}$ defined as $\mathrm{AE}\left(\widehat{\vc x}\right):=\frac{1}{\pi}\cos^{-1}\left\langle \widehat{\vc x},\vc x^{\star}\right\rangle $.
As can be seen from the figure, with higher input SNR (i.e., $\eta$)
and less sparse target signals the algorithms incur larger angular
error. While there is no significant difference in performance of
GaSP, GraSP-$\eta$, and BIHT-$\ell_{2}$ for the examined values
of $\eta$ and $s$, the BIHT algorithm appears to be sensitive to
$\eta$. At $\eta=20$dB and low sampling ratios BIHT outperforms
the other methods by a noticeable margin. However, for more noisy
measurements BIHT loses its advantage and at $\eta=0$dB it performs
even poorer than the PV-$\ell_{0}$. PV-$\ell_{0}$ never outperforms
the two variants of GraSP or the BIHT-$\ell_{2}$, but the gap between
their achieved angular error decreases as the measurements become
more noisy.

\begin{sidewaysfigure*}
\centering\includegraphics[width=1\textwidth]{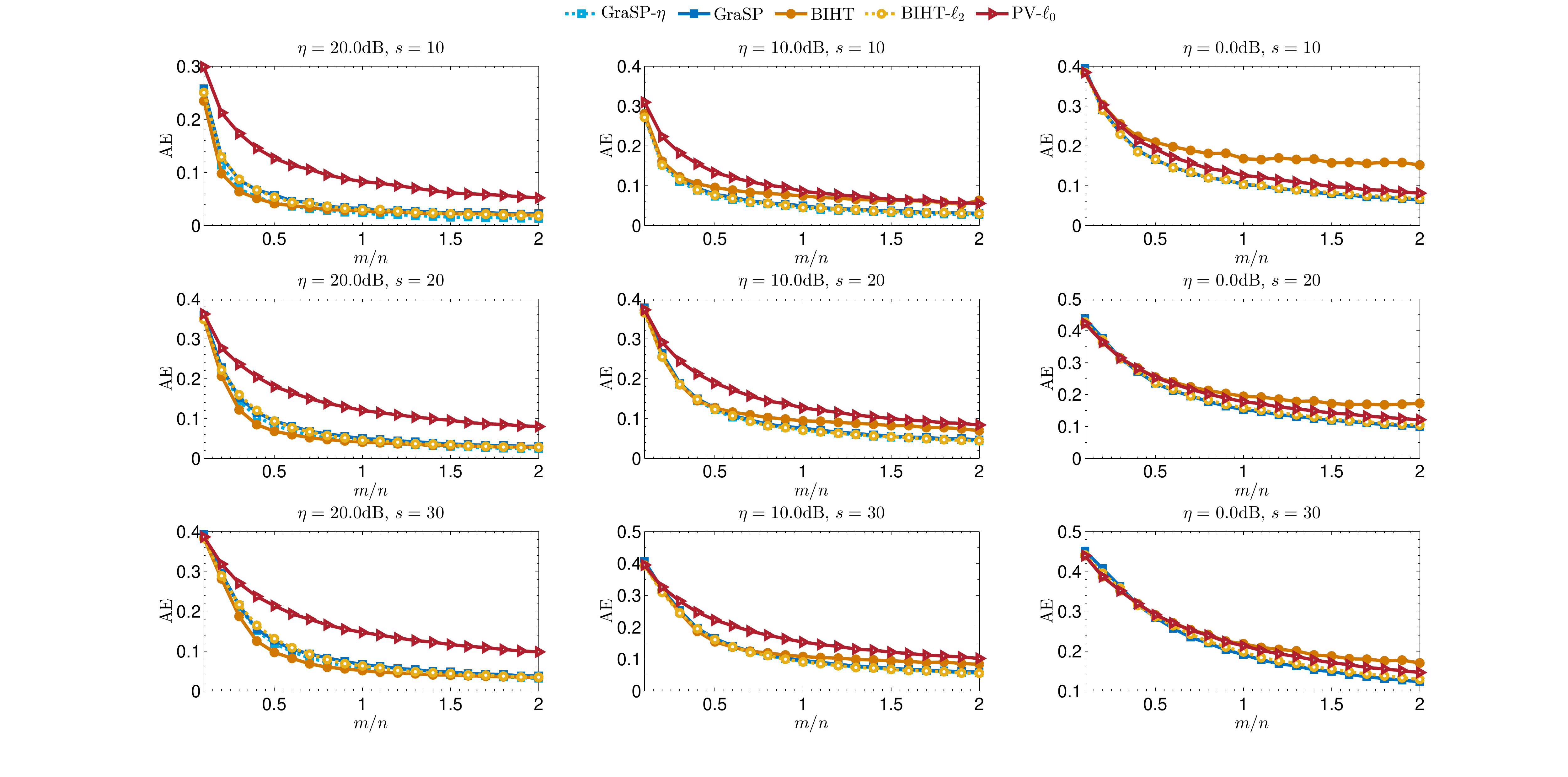}\caption{Angular error (AE) vs. the sampling ratio ($m/n$) at different values
of input SNR ($\eta$) and sparsity ($s$)\label{fig:AE}}
\end{sidewaysfigure*}

The reconstruction SNR of the estimates produced by the algorithms
are compared in Figure \ref{fig:RSNR}. The reconstruction SNR conveys
the same information as the angular error as it can be calculated
through the formula 
\begin{alignat*}{1}
\mathrm{R-SNR}\left(\widehat{\vc x}\right) & :=-20\log_{10}\norm[2]{\widehat{\vc x}-\vc x^{\star}}\\
 & =-10\log_{10}\left(2-2\cos\mathrm{AE}\left(\widehat{\vc x}\right)\right).
\end{alignat*}
 However, it magnifies small differences between the algorithms that
were difficult to trace using the angular error. For example, it can
be seen in Figure \ref{fig:RSNR} that at $\eta=20$dB and $s=10$,
GraSP-$\eta$ has an advantage (of up to 2dB) in reconstruction SNR.

\begin{sidewaysfigure*}
\centering\includegraphics[width=1\textwidth]{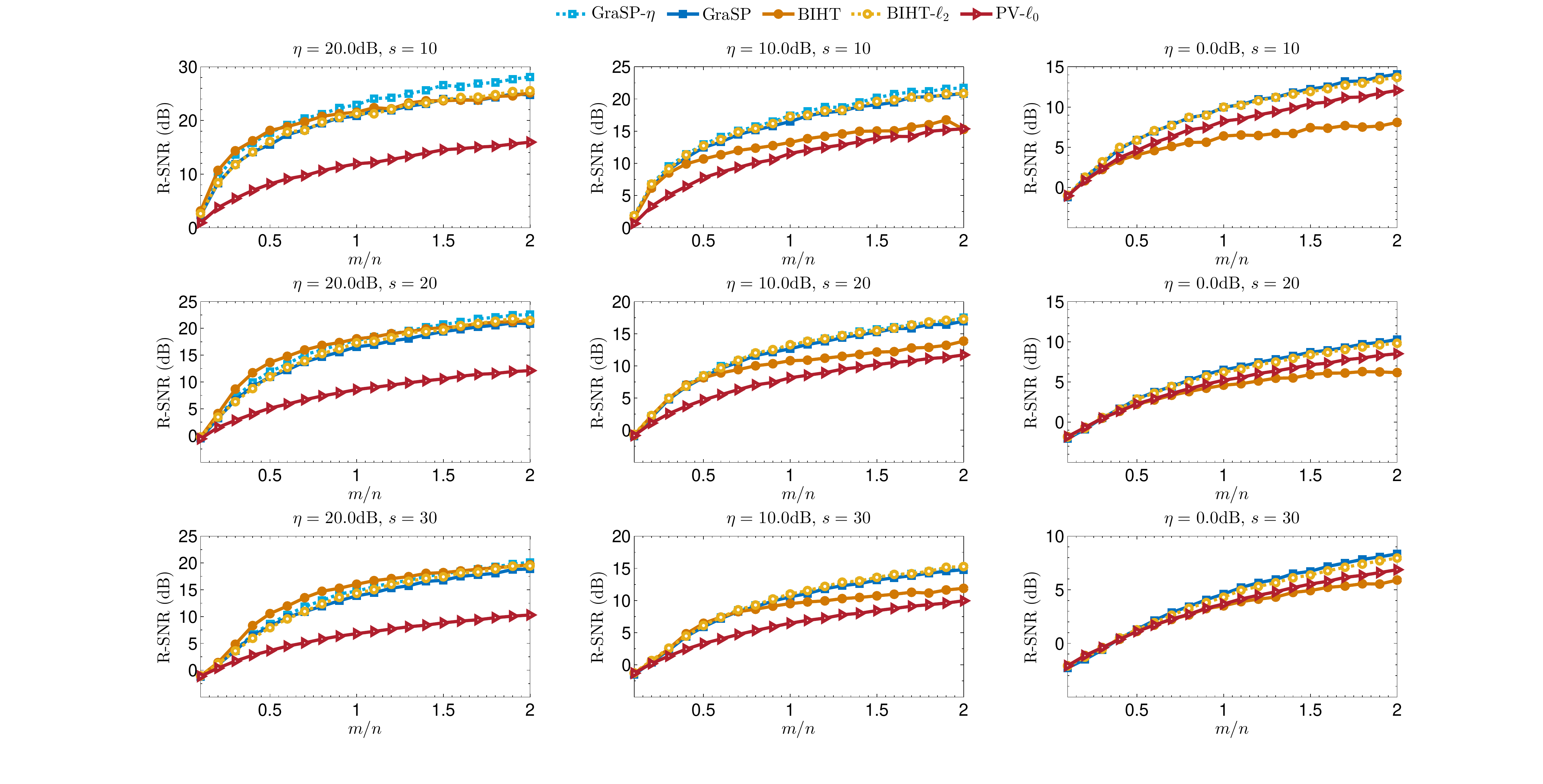}\caption{Reconstruction SNR on the unit ball (R-SNR) vs. the sampling ratio
($m/n$) at different values of input SNR ($\eta$) and sparsity ($s$)\label{fig:RSNR}}
\end{sidewaysfigure*}

Furthermore, we evaluated performance of the algorithms in terms of
identifying the correct support set of the target sparse signal by
are comparing their achieved False Negative Rate 
\begin{alignat*}{1}
\mathrm{FNR} & =\frac{\left|\supp\left(\vc x^{\star}\right)\backslash\supp\left(\widehat{\vc x}\right)\right|}{\left|\supp\left(\vc x^{\star}\right)\right|}
\end{alignat*}
 and False Positive Rate 
\begin{alignat*}{1}
\mathrm{FPR} & =\frac{\left|\supp\left(\widehat{\vc x}\right)\backslash\supp\left(\vc x^{\star}\right)\right|}{n-\left|\supp\left(\vc x^{\star}\right)\right|}.
\end{alignat*}
Figures \ref{fig:FNR} and \ref{fig:FPR} illustrate these rates for
the studied algorithms. It can be seen in Figure \ref{fig:FNR} that
at $\eta=20$dB, BIHT achieves a FNR slightly lower than that of the
variants of GraSP, whereas PV-$\ell_{0}$ and BIHT-$\ell_{2}$ rank
first and second, respectively, in the highest FNR at a distant from
the other algorithms. However, as $\eta$ decreases the FNR of BIHT
deteriorates relative to the other algorithms while BIHT-$\ell_{2}$
shows improved FNR. The GraSP variants exhibit better performance
overall at smaller values of $\eta$ especially with $s=10$, but
for $\eta=10$dB and at low sampling ratios BIHT attains a slightly
better FNR. The relative performance of the algorithms in terms of
FPR, illustrated in Figure \ref{fig:FPR}, is similar. 

\begin{sidewaysfigure*}
\centering\includegraphics[width=1\textwidth]{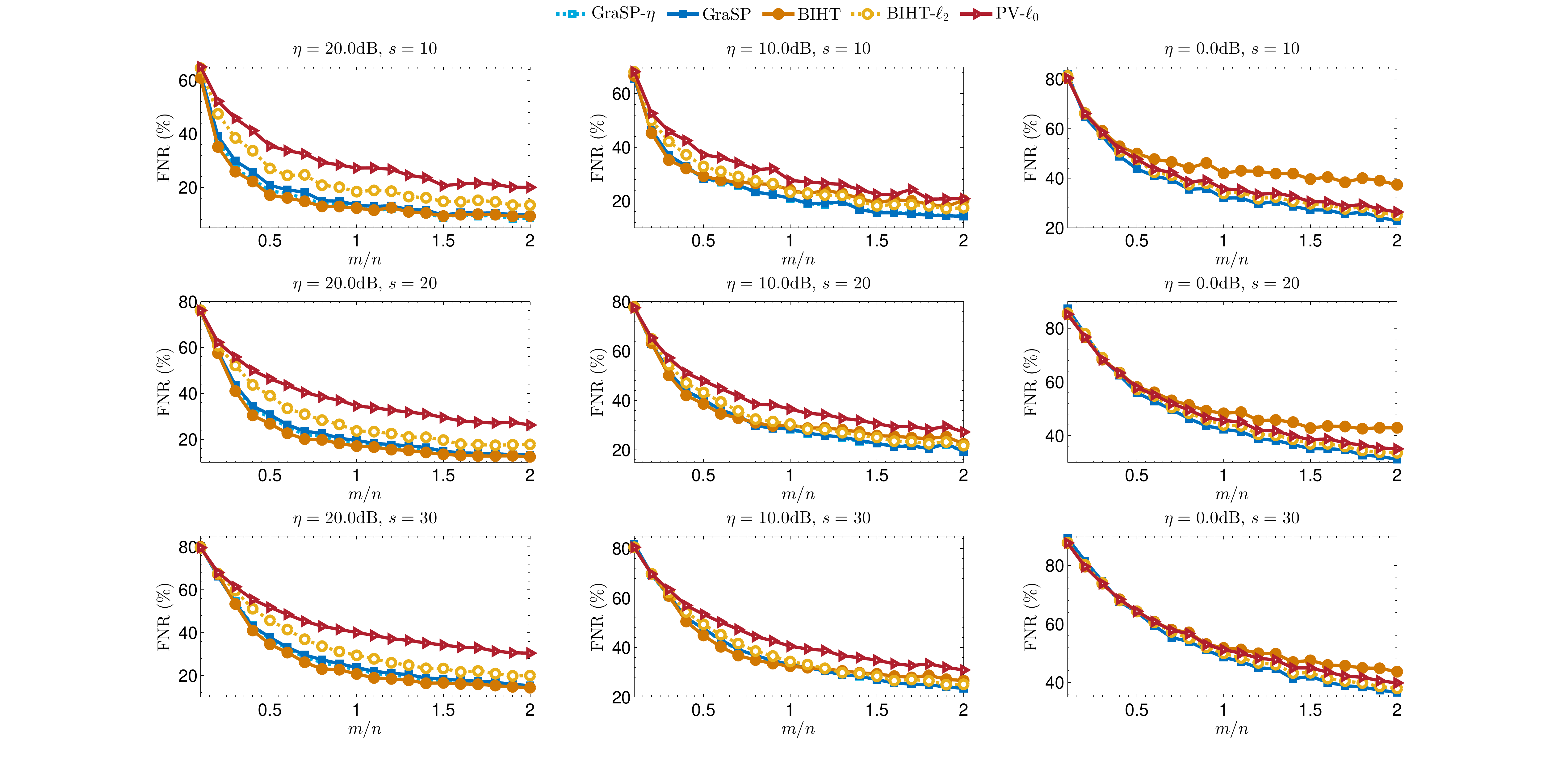}\caption{False Negative Rate (FNR) vs. the sampling ratio ($m/n$) at different
values of input SNR ($\eta$) and sparsity~($s$)\label{fig:FNR}}
\end{sidewaysfigure*}
\begin{sidewaysfigure*}
\centering\includegraphics[width=1\textwidth]{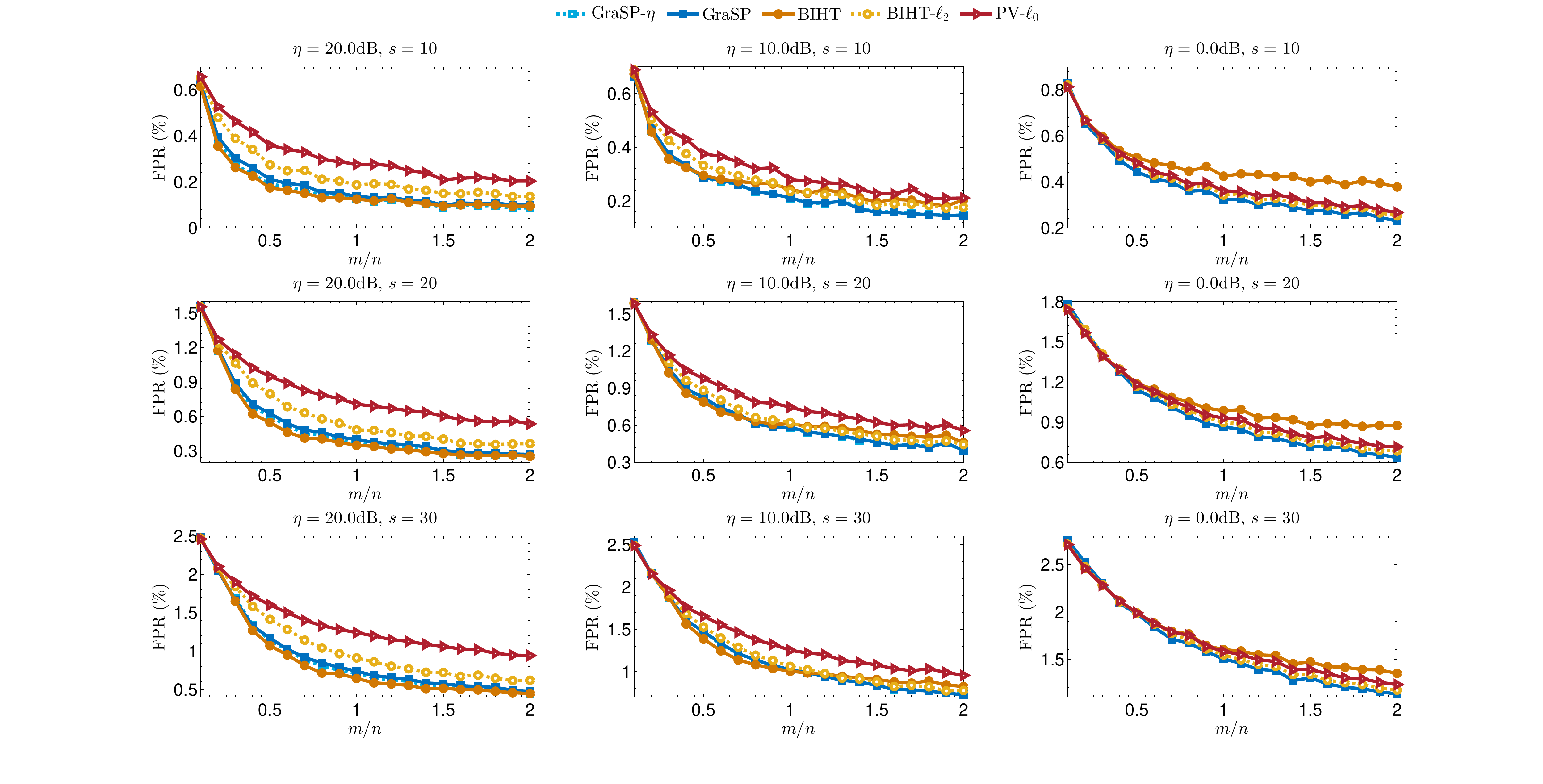}\caption{False Positive Rate (FPR) vs. the sampling ratio ($m/n$) at different
values of input SNR ($\eta$) and sparsity~($s$)\label{fig:FPR}}
\end{sidewaysfigure*}

We also compared the algorithms in terms of their average execution
time ($T$) measured in seconds. The simulation was ran on a PC with
an AMD Phenom\texttrademark II X6 2.60GHz processor and 8.00GB of
RAM. The average execution time of the algorithms, all of which are
implemented in MATLAB\textsuperscript{\textregistered}, is illustrated
in \ref{fig:ET} in log scale. It can be observed from the figure
that PV-$\ell_{0}$ is the fastest algorithm which can be attributed
to its non-iterative procedure. Furthermore, in general BIHT-$\ell_{2}$
requires significantly longer time compared to the other algorithms.
The BIHT, however, appears to be the fastest among the iterative algorithms
at low sampling ratio or at large values of $\eta$. The GraSP variants
generally run at similar speed, while they are faster than BIHT at
low values of $\eta$ and high sampling ratios.

\begin{sidewaysfigure*}
\centering\includegraphics[width=1\textwidth]{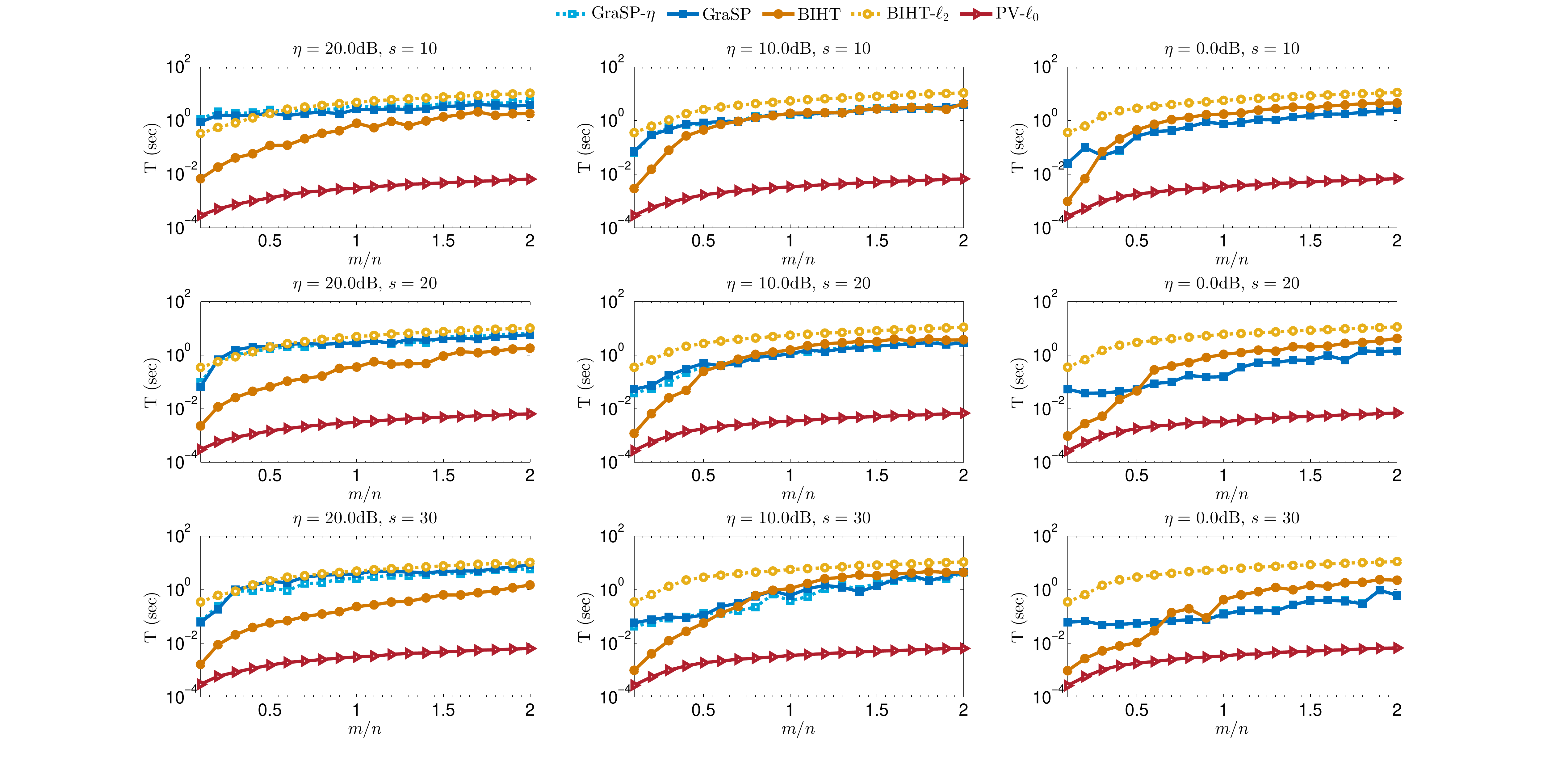}\caption{Average execution time (T) vs. the sampling ratio ($m/n$) at different
values of input SNR ($\eta$) and sparsity~($s$)\label{fig:ET}}
\end{sidewaysfigure*}

\section{Conclusion\label{sec:Conclusion}}

In this paper we revisited a formulation of the 1-bit CS problem based
on the maximum likelihood estimation. Furthermore, we applied a variant
of the GraSP algorithm \cite{Bahmani_Greedy_2012} to this problem.
We showed through numerical simulations that the proposed algorithms
have robust performance in presence of noise. While at high levels
of input SNR these algorithms are outperformed by a narrow margin
by the competing algorithms, in low input SNR regime our algorithms
show a solid performance at reasonable computational cost.

\appendices

\section{Proofs\label{apndx:1-bitCS}}

To prove Theorem \ref{thm:IterationInvariant} we use the following
two lemmas. We omit the proofs since they can be easily adapted from
Lemmas 1 and 2 in \cite{Bahmani_Greedy_2012} using straightforward
changes. It suffices to notice that 
\begin{enumerate}
\item the proof in \cite{Bahmani_Greedy_2012} still holds if the estimation
errors are measured with respect to the true sparse minimizer or any
other feasible (i.e., $s$-sparse) point, rather than the statistical
true parameter, and
\item the iterates and the crude estimates will always remain in the sphere
of radius $r$ centered at the origin where the SRH applies.
\end{enumerate}
In what follows $\int_{0}^{1}\alpha_{k}\left(\tau\vc x+\left(1\!-\!\tau\right)\overline{\vc x}\right)\dx\tau$
and $\int_{0}^{1}\beta_{k}\left(\tau\vc x+\left(1\!-\!\tau\right)\overline{\vc x}\right)\dx\tau$
are denoted by $\widetilde{\alpha}_{k}\left(\vc x\right)$ and $\widetilde{\beta}_{k}\left(\vc x\right)$,
respectively. We also define $\widetilde{\gamma}_{k}\left(\vc x\right):=\widetilde{\alpha}_{k}\left(\vc x\right)-\widetilde{\beta}_{k}\left(\vc x\right)$.
\begin{lem}
\label{lem:VComplement}Let $\st Z$ be the index set defined in Algorithm
\ref{alg:1bit-GraSP} and $\st R$ denote the set $\supp\left(\vc x\itr t-\overline{\vc x}\right)$.
Then the iterate $\vc x\itr t$ obeys 
\begin{alignat*}{1}
\norm[2]{\res{\left(\vc x\itr t-\overline{\vc x}\right)}_{\st Z\cmpl}} & \leq\frac{\widetilde{\gamma}_{4s}\left(\vc x\itr t\right)+\widetilde{\gamma}_{2s}\left(\vc x\itr t\right)}{\widetilde{\beta}_{2s}\left(\vc x\itr t\right)}\norm[2]{\vc x\itr t-\overline{\vc x}}+\frac{\norm[2]{\nabla_{\st R\backslash\st Z}f\left(\overline{\vc x}\right)}+\norm[2]{\nabla_{\st Z\backslash\st R}f\left(\overline{\vc x}\right)}}{\widetilde{\beta}_{2s}\left(\vc x\itr t\right)}.\\
\end{alignat*}

\end{lem}

\begin{lem}
\label{lem:ErrorWRTb}The vector $\vc b$ defined at line \ref{lin:crude}
of Algorithm \ref{alg:1bit-GraSP} obeys 
\begin{alignat*}{1}
\norm[2]{\res{\overline{\vc x}}_{\st T}-\vc b} & \leq\frac{\norm[2]{\nabla_{\st T}f\left(\overline{\vc x}\right)}}{\widetilde{\beta}_{4s}\left(\vc b\right)}+\frac{\widetilde{\gamma}_{4s}\left(\vc b\right)}{2\widetilde{\beta}_{4s}\left(\vc b\right)}\norm[2]{\res{\overline{\vc x}}_{\st T\cmpl}}.
\end{alignat*}
\end{lem}
\begin{proof}[\textbf{Proof of Theorem \ref{thm:IterationInvariant}}]
Since $\st Z\subseteq\st T$ we have $\st T\cmpl\subseteq\st Z\cmpl$
and thus 
\begin{alignat*}{1}
\norm[2]{\res{\left(\vc x\itr t-\overline{\vc x}\right)}_{\st Z\cmpl}} & \geq\norm[2]{\res{\left(\vc x\itr t-\overline{\vc x}\right)}_{\st T\cmpl}}\\
 & =\norm[2]{\res{\overline{\vc x}}_{\st T\cmpl}}.
\end{alignat*}
 Then it follows from Lemma \ref{lem:VComplement} that 
\begin{alignat}{1}
\norm[2]{\res{\overline{\vc x}}_{\st T\cmpl}} & \leq\frac{\widetilde{\gamma}_{4s}\left(\vc x\itr t\right)}{\widetilde{\beta}_{4s}\left(\vc x\itr t\right)}\norm[2]{\vc x\itr t-\overline{\vc x}}\nonumber \\
 & +\frac{\norm[2]{\nabla_{\st R\backslash\st Z}f\left(\overline{\vc x}\right)}+\norm[2]{\nabla_{\st Z\backslash\st R}f\left(\overline{\vc x}\right)}}{\beta_{4s}}\nonumber \\
 & \leq\left(\mu_{4s}-1\right)\norm[2]{\vc x\itr t-\overline{\vc x}}+2\epsilon,\label{eq:xTc}
\end{alignat}
 where we used the fact that $\alpha_{4s}\geq\alpha_{2s}$ and $\beta_{4s}\leq\beta_{2s}$
to simplify the expressions. Furthermore, we have 
\begin{alignat*}{1}
\norm[2]{\vc x\itr{t+1}-\overline{\vc x}} & =\norm[2]{\vc b_{s}-\overline{\vc x}}\\
 & \leq\norm[2]{\vc b_{s}-\res{\overline{\vc x}}_{\st T}}+\norm[2]{\res{\overline{\vc x}}_{\st T\cmpl}}\\
 & \leq\norm[2]{\vc b_{s}-\vc b}+\norm[2]{\vc b-\res{\overline{\vc x}}_{\st T}}+\norm[2]{\res{\overline{\vc x}}_{\st T\cmpl}}\\
 & \leq2\norm[2]{\vc b-\res{\overline{\vc x}}_{\st T}}+\norm[2]{\res{\overline{\vc x}}_{\st T\cmpl}},
\end{alignat*}
 where the last inequality holds because $\vc b_{s}$ is the best
$s$-term approximation of $\vc b$. Hence, it follows from Lemma
\ref{lem:ErrorWRTb} that 
\begin{alignat*}{1}
\norm[2]{\vc x\itr{t+1}-\overline{\vc x}} & \leq2\frac{\norm[2]{\nabla_{\st T}f\left(\overline{\vc x}\right)}}{\widetilde{\beta}_{4s}\left(\vc b\right)}+\frac{\widetilde{\alpha}_{4s}\left(\vc b\right)}{\widetilde{\beta}_{4s}\left(\vc b\right)}\norm[2]{\res{\overline{\vc x}}_{\st T\cmpl}}\\
 & \leq2\epsilon+\mu_{4s}\norm[2]{\res{\overline{\vc x}}_{\st T\cmpl}}.
\end{alignat*}
 Then applying (\ref{eq:xTc}) and simplifying the resulting inequality
yield 
\begin{alignat*}{1}
\norm[2]{\vc x\itr{t+1}-\overline{\vc x}} & \leq2\epsilon+\mu_{4s}\left(\left(\mu_{4s}-1\right)\norm[2]{\vc x\itr t-\overline{\vc x}}+2\epsilon\right)\\
 & \leq\left(\mu_{4s}^{2}-\mu_{4s}\right)\norm[2]{\vc x\itr t-\overline{\vc x}}+2\left(\mu_{4s}+1\right)\epsilon,
\end{alignat*}
which is the desired result.\end{proof}
\begin{lem}[Bounded Sparse Projection]
\label{lem:BSparseProj}For any $\vc x\in\mathbb{R}^{n}$ the vector
$\max\left\{ 1,\frac{r}{\norm[2]{\vc x_{s}}}\right\} \vc x_{s}$ is
a solution to the minimization
\begin{alignat}{1}
\arg\min_{\vc w}\ \norm[2]{\vc x-\vc w} & \quad\text{s.t. }\norm[2]{\vc w}\leq r\text{ and }\norm[0]{\vc w}\leq s.\label{eq:BSProj}
\end{alignat}
\end{lem}
\begin{proof}
Given an index set $\st S\subseteq\left[n\right]$ we can write $\norm[2]{\vc x-\vc w}^{2}=\norm[2]{\res{\vc x-\vc w}_{\st S}}^{2}+\norm[2]{\res{\vc x}_{\st S\cmpl}}^{2}$
for vectors $\vc w$ with $\supp\left(\vc w\right)\subseteq\st S$.
Therefore, the solution to 
\begin{alignat*}{1}
\arg\min_{\vc w}\ \norm[2]{\vc x-\vc w} & \quad\text{s.t. }\norm[2]{\vc w}\leq r\text{ and }\supp\left(\vc w\right)\subseteq\st S
\end{alignat*}
 is simply obtained by projection of $\res{\vc x}_{\st S}$ onto the
sphere of radius $r$, i.e., 
\begin{alignat*}{1}
\mathrm{P}_{\st S}\left(\vc x\right) & =\max\left\{ 1,\frac{r}{\norm[2]{\res{\vc x}_{\st S}}}\right\} \res{\vc x}_{\st S}.
\end{alignat*}
 Therefore, to find a solution to (\ref{eq:BSProj}) it suffices to
find the index set $\st S$ with $\left|\st S\right|=s$ and thus
the corresponding $\mathrm{P}_{\st S}\left(\vc x\right)$ that minimize
$\norm[2]{\vc x-\mathrm{P}_{\st S}\left(\vc x\right)}$. Note that
we have 
\begin{alignat*}{1}
\norm[2]{\vc x-\mathrm{P}_{\st S}\left(\vc x\right)}^{2} & =\norm[2]{\res{\vc x}_{\st S}-\mathrm{P}_{\st S}\left(\vc x\right)}^{2}+\norm[2]{\res{\vc x}_{\st S\cmpl}}^{2}\displaybreak[0]\\
 & =\left(\norm[2]{\res{\vc x}_{\st S}}-r\right)_{+}^{2}+\norm[2]{\res{\vc x}_{\st S\cmpl}}^{2}\\
 & =\begin{cases}
\norm[2]{\vc x}^{2}-\norm[2]{\res{\vc x}_{\st S}}^{2} & \ ,\norm[2]{\res{\vc x}_{\st S}}<r\\
\norm[2]{\vc x}^{2}+r^{2}-2r\norm[2]{\res{\vc x}_{\st S}} & \ ,\norm[2]{\res{\vc x}_{\st S}}\geq r
\end{cases}.
\end{alignat*}
For all valid $\st S$ with $\norm[2]{\res{\vc x}_{\st S}}<r$ we
have $\norm[2]{\vc x}^{2}-\norm[2]{\res{\vc x}_{\st S}}^{2}>\norm[2]{\vc x}^{2}-r^{2}$.
Similarly, for all valid $\st S$ with $\norm[2]{\res{\vc x}_{\st S}}<r$
we have $\norm[2]{\vc x}^{2}+r^{2}-2r\norm[2]{\res{\vc x}_{\st S}}\leq\norm[2]{\vc x}^{2}-r^{2}$.
Furthermore, both $\norm[2]{\vc x}^{2}-\norm[2]{\res{\vc x}_{\st S}}^{2}$
and $\norm[2]{\vc x}^{2}+r^{2}-2r\norm[2]{\res{\vc x}_{\st S}}$ are
decreasing functions of $\norm[2]{\res{\vc x}_{\st S}}$. Therefore,
$\norm[2]{\vc x-\mathrm{P}_{\st S}\left(\vc x\right)}^{2}$ is a decreasing
function of $\norm[2]{\res{\vc x}_{\st S}}$. Hence, $\norm[2]{\vc x-\mathrm{P}_{\st S}\left(\vc x\right)}$
attains its minimum at $\st S=\supp\left(\vc x_{s}\right)$.
\end{proof}

\section{On non-convex formulation of \cite{Plan_Robust_2013}\label{apndx:PV0}}

\cite{Plan_Robust_2013} derived accuracy guarantees for 
\begin{alignat*}{1}
\arg\max_{\vc x}\ \left\langle \vc y,\mx A\vc x\right\rangle  & \quad\text{s.t. }\mbox{\ensuremath{\vc x\in\st K}}
\end{alignat*}
 as a solver for the 1-bit CS problem, where $\st K$ is a subset
of the unit Euclidean ball. While their result \cite[Theorem 1.1]{Plan_Robust_2013}
applies to both convex and non-convex sets $\st K$, the focus of
their work has been on the set $\st K$ that is the intersection of
a centered $\ell_{1}$-ball and the unit Euclidean ball. Our goal,
however, is to examine the other interesting choice of $\st K$, namely
the intersection of canonical sparse subspaces and the unit Euclidean
ball. The estimator in this case can be written as 
\begin{alignat}{1}
\arg\max_{\vc x}\ \left\langle \vc y,\mx A\vc x\right\rangle  & \quad\text{s.t. }\norm[0]{\vc x}\leq s\text{ and }\norm[2]{\vc x}\leq1.\label{eq:PVL0}
\end{alignat}
 We show that a solution to the optimization above can be obtained
explicitly.
\begin{lem}
A solution to (\ref{eq:PVL0}) is $\widehat{\vc x}=\left(\mx A\tran\vc y\right)_{s}/\norm[2]{\left(\mx A\tran\vc y\right)_{s}}$.\end{lem}
\begin{proof}
For $\st I\subseteq\left[n\right]$ define 
\begin{alignat*}{1}
\widehat{\vc x}\left(\st I\right) & :=\arg\max_{\mx x}\ \left\langle \vc y,\mx A\vc x\right\rangle \quad\text{s.t. }\res{\vc x}_{\st I\cmpl}=0\text{ and }\norm[2]{\vc x}\leq1.
\end{alignat*}
Furthermore, choose 
\begin{alignat*}{1}
\widehat{\st I} & \in\arg\max_{\st I}\ \left\langle \vc y,\mx A\widehat{\vc x}\left(\st I\right)\right\rangle \quad\text{s.t. }\st I\subseteq\left[n\right]\text{ and }\left|\st I\right|\leq s.
\end{alignat*}
Then $\widehat{\vc x}\left(\widehat{\st I}\right)$ would be a solution
to (\ref{eq:PVL0}). Using the fact that $\left\langle \vc y,\mx A\vc x\right\rangle =\left\langle \mx A\tran\vc y,\vc x\right\rangle $,
straightforward application of Cauchy-Schwarz inequality shows that
$\widehat{\vc x}\left(\st I\right)=\res{\left(\mx A\tran\vc y\right)}_{\st I}/\norm[2]{\res{\left(\mx A\tran\vc y\right)}_{\st I}}$
for which we have 
\begin{alignat*}{1}
\left\langle \vc y,\mx A\widehat{\vc x}\left(\st I\right)\right\rangle  & =\norm[2]{\res{\left(\mx A\tran\vc y\right)}_{\st I}}.
\end{alignat*}
 Thus, we obtain $\widehat{\st I}=\supp\left(\left(\mx A\tran\vc y\right)_{s}\right)$
and thereby $\widehat{\vc x}\left(\widehat{\st I}\right)=\widehat{\vc x}$,
which proves the claim. 
\end{proof}
\bibliographystyle{IEEEtran}
\bibliography{references}

\end{document}

%% file: macros.tex
\global\long\def\bbone{1\hspace{-0.6ex}\mathbf{l}}
\global\long\def\vc#1{\mathbf{#1}}
\global\long\def\mx#1{\mathbf{#1}}
\global\long\def\st#1{\mathcal{#1}}
\newcommand\norms[2][]{%
{\ifx&#1&%
#2
\else
#2_{#1}
\fi}
}\newcommandx\norm[2][usedefault, addprefix=\global, 1=]{\norms[#1]{\left\Vert #2\right\Vert }}
\global\long\def\res#1{\left.#1\right|}
\global\long\def\herm{^{\mspace{-1mu}\mathsf{H}}}
\global\long\def\cmpl{^{^{\mathrm{c}}}}
\global\long\def\tran{^{\mathrm{T}}}
\global\long\def\dx{\mathrm{d}}
\global\long\def\Arg{\text{Arg}}
\global\long\def\supp{\mathrm{supp}}
\global\long\def\vcg#1{\boldsymbol{#1}}
\global\long\def\tr{\mathrm{tr}}
\global\long\def\mxg#1{\boldsymbol{#1}}
\global\long\def\itr#1{^{\left(#1\right)}}
\global\long\def\Breg#1#2#3{\mathrm{B}_{#1}\left(#2\parallel#3\right)}
\global\long\def\bsym#1{\boldsymbol{#1}}
\global\long\def\sgn{\mathrm{sgn}}

%% file: manuscript.bbl
\begin{thebibliography}{10}
\providecommand{\url}[1]{#1}
\csname url@samestyle\endcsname
\providecommand{\newblock}{\relax}
\providecommand{\bibinfo}[2]{#2}
\providecommand{\BIBentrySTDinterwordspacing}{\spaceskip=0pt\relax}
\providecommand{\BIBentryALTinterwordstretchfactor}{4}
\providecommand{\BIBentryALTinterwordspacing}{\spaceskip=\fontdimen2\font plus
\BIBentryALTinterwordstretchfactor\fontdimen3\font minus
  \fontdimen4\font\relax}
\providecommand{\BIBforeignlanguage}[2]{{%
\expandafter\ifx\csname l@#1\endcsname\relax
\typeout{** WARNING: IEEEtran.bst: No hyphenation pattern has been}%
\typeout{** loaded for the language `#1'. Using the pattern for}%
\typeout{** the default language instead.}%
\else
\language=\csname l@#1\endcsname
\fi
#2}}
\providecommand{\BIBdecl}{\relax}
\BIBdecl

\bibitem{Laska_Finite_2009}
J.~N. Laska, P.~T. Boufounos, and R.~G. Baraniuk, ``Finite range scalar
  quantization for compressive sensing,'' in \emph{Proceedings of International
  Conferenece on Sampling Theory and Applications (SampTA)}, Toulouse, France,
  May 18-22 2009.

\bibitem{Dai_Distortion-rate_2009}
W.~Dai, H.~V. Pham, and O.~Milenkovic, ``Distortion-rate functions for
  quantized compressive sensing,'' in \emph{{IEEE} Information Theory Workshop
  on Networking and Information Theory, 2009. {ITW} 2009}, Jun. 2009, pp. 171
  --175.

\bibitem{Sun_Optimal_2009}
J.~Sun and V.~Goyal, ``Optimal quantization of random measurements in
  compressed sensing,'' in \emph{{IEEE} International Symposium on Information
  Theory, 2009. {ISIT} 2009}, Jul. 2009, pp. 6 --10.

\bibitem{Zymnis_Compressed_2010}
A.~Zymnis, S.~Boyd, and E.~Cand\`{e}s, ``Compressed sensing with quantized
  measurements,'' \emph{{IEEE} Signal Processing Letters}, vol.~17, no.~2, pp.
  149--152, Feb. 2010.

\bibitem{Jacques_Dequantizing_2011}
L.~Jacques, D.~Hammond, and J.~Fadili, ``Dequantizing compressed sensing: when
  oversampling and non-gaussian constraints combine,'' \emph{{IEEE}
  Transactions on Information Theory}, vol.~57, no.~1, pp. 559--571, Jan. 2011.

\bibitem{Laska_Democracy_2011}
J.~N. Laska, P.~T. Boufounos, M.~A. Davenport, and R.~G. Baraniuk, ``Democracy
  in action: Quantization, saturation, and compressive sensing,'' \emph{Applied
  and Computational Harmonic Analysis}, vol.~31, no.~3, pp. 429--443, Nov.
  2011.

\bibitem{Boufounos_1-bit_2008}
P.~Boufounos and R.~Baraniuk, ``1-bit compressive sensing,'' in
  \emph{Information Sciences and Systems, 2008. {CISS} 2008. 42nd Annual
  Conference on}, Mar. 2008, pp. 16 --21.

\bibitem{Plan_Robust_2013}
Y.~Plan and R.~Vershynin, ``Robust 1-bit compressed sensing and sparse logistic
  regression: A convex programming approach,'' \emph{{IEEE} Transactions on
  Information Theory}, vol.~59, no.~1, pp. 482--494, Jan. 2013.

\bibitem{Boufounos_Greedy_2009}
P.~Boufounos, ``Greedy sparse signal reconstruction from sign measurements,''
  in \emph{Conference Record of the Forty-Third Asilomar Conference on Signals,
  Systems and Computers, 2009}, Nov. 2009, pp. 1305 --1309.

\bibitem{Needell_CoSaMP_2009}
D.~Needell and J.~A. Tropp, ``{CoSaMP}: Iterative signal recovery from
  incomplete and inaccurate samples,'' \emph{Applied and Computational Harmonic
  Analysis}, vol.~26, no.~3, pp. 301--321, 2009.

\bibitem{Laska_Trust_2011}
J.~Laska, Z.~Wen, W.~Yin, and R.~Baraniuk, ``Trust, but verify: Fast and
  accurate signal recovery from 1-bit compressive measurements,'' \emph{{IEEE}
  Transactions on Signal Processing}, vol.~59, no.~11, pp. 5289--5301, Nov.
  2011.

\bibitem{Jacques_Robust_2013}
L.~Jacques, J.~Laska, P.~Boufounos, and R.~Baraniuk, ``Robust 1-bit compressive
  sensing via binary stable embeddings of sparse vectors,'' \emph{{IEEE}
  Transactions on Information Theory}, vol.~PP, 2013,
  \href{http://arxiv.org/abs/1104.3160}{\tt arXiv:1104.3160 [cs.IT]}.

\bibitem{Yan_Robust_2012}
M.~Yan, Y.~Yang, and S.~Osher, ``Robust 1-bit compressive sensing using
  adaptive outlier pursuit,'' \emph{{IEEE} Transactions on Signal Processing},
  vol.~60, no.~7, pp. 3868--3875, Jul. 2012.

\bibitem{Plan_One-bit_2011}
Y.~Plan and R.~Vershynin, ``One-bit compressed sensing by linear programming,''
  Sep. 2011, \href{http://arxiv.org/abs/1109.4299}{\tt arxiv:1109.4299
  [cs.IT]}.

\bibitem{Ai_One-bit_NG_2013}
A.~Ai, A.~Lapanowski, P.~Yaniv, and R.~Vershynin, ``One-bit compressed sensing
  with non-gaussian measurements,'' \emph{Linear Algebra and its Applications},
  2013, to appear.

\bibitem{Bahmani_Greedy_2012}
S.~Bahmani, B.~Raj, and P.~Boufounos, ``Greedy sparsity-constrained
  optimization,'' \emph{Journal of Machine Learning Research}, vol.~14, no.~3,
  pp. 807--841, Mar. 2013.

\end{thebibliography}
